\documentclass[a4paper,12pt]{article}
\usepackage{amssymb,amsmath,amsthm}
\usepackage{bm}%
\usepackage{ascmac}
\usepackage{mathrsfs}
\usepackage{stmaryrd}
\usepackage{arydshln}
\usepackage{comment}
\usepackage{braket}
\usepackage{xcolor}
	\colorlet{mycol}{black}
\usepackage{upgreek}
\usepackage{hyperref}
\hypersetup{
	colorlinks=true,
	citecolor=blue,
	linkcolor=blue,
	urlcolor=blue,
}

\def\nn{\notag}
\def\Z2{\mathbb{Z}_2^2}
\def\g{\mathfrak{g}}

\def\DP#1#2{\hat{#1}\cdot \hat{#2}}
\def\PH#1#2{(-1)^{\hat{#1}\cdot \hat{#2}}}
\def\osp{\mathfrak{osp}}

\def\pr#1{\overleftarrow{\mathbf{V}}^{\mathrm{pr}}_{#1}}
\def\FD#1#2{\mathcal{D}_{#1}^{[#2]}}

\numberwithin{equation}{section} 
\newtheorem{dfn}{Definition}
\newtheorem{thm}{Theorem}
\newtheorem{prop}{Proposition}
\newtheorem{lemma}{Lemma}

\theoremstyle{definition}

\begin{document}
	
\thispagestyle{empty}
\begin{center}
	
	{\Large\bfseries
		An integrable $\mathbb{Z}_2^2$-graded extension of Camassa-Holm equation and its bi-Hamiltonian structure
	}
	
	\vspace{1.5cm}
	
	{\large
		N. Aizawa,\textsuperscript{1} 
		\quad 
		Ichi Fujii,\textsuperscript{2} 
		\quad 
		Ren Ito\textsuperscript{3} and \quad
		L. Vasconcellos\textsuperscript{4}
	}
	
	\vspace{5mm}
	
	{\small
		\textsuperscript{1,2,3}Department of Physics, Graduate School of Science, \\
		Osaka Metropolitan University, Sugimoto Campus,
		Osaka 599-8585, Japan\\

		\vspace{3mm}
		
		\textsuperscript{4}Department of Mathematics, Universidade Federal de S\~{a}o Carlos (UFSCar),\\
		S\~{a}o Carlos - SP, Brazil\\
		
		\vspace{3mm}
		\textsuperscript{1}\texttt{aizawa@omu.ac.jp} \\
		\textsuperscript{2}\texttt{sq25482d@st.omu.ac.jp}\\
		\textsuperscript{3}\texttt{sd22709y@st.omu.ac.jp}\\
		\textsuperscript{4}\texttt{lucas.vasconcellos@estudante.ufscar.br}
	}
	
	\vspace{1.5cm}
\end{center}

\vfill
\begin{abstract}
	By constructing Lax operators in the loop algebra of the $\mathbb{Z}_2^2$-graded extension of the Lie superalgebra $\osp(1|2)$, we derive a $\mathbb{Z}_2^2$-graded extension of the Camassa-Holm equation. The resulting equation is an integrable nonlinear PDE for a system of four $\mathbb{Z}_2^2$-graded commutative functions, each associated with a distinct $\mathbb{Z}_2^2$-degree. We further show that the $\mathbb{Z}_2^2$-Camassa-Holm equation admits a bi-Hamiltonian structure. 
	As a consequence, it  possesses infinitely many conserved quantities, including one with non-trivial $\mathbb{Z}_2^2$-degree, which are mutually in involution with respect to the $\mathbb{Z}_2^2$-graded Poisson brackets.

\end{abstract}

\clearpage
\setcounter{page}{1}
\section{Introduction}

Classical integrable systems such as soliton equations are deeply connected to Lie algebras. 
Hierarchies of integrable equations can be constructed from data associated with a Lie algebra. 
By replacing the Lie algebra with a Lie superalgebra, one may generalize integrable systems to include both bosonic and fermionic degrees of freedom. 
It has been demonstrated that a further generalization is possible by replacing the usual $\mathbb Z_2$-grading with a $\Z2$-grading, leading to $\Z2$-graded Lie (super)algebras, where $\Z2:=\mathbb Z_2\times\mathbb Z_2$ \cite{BruSG,AIKTTslint,AiFujiIto,AFITTsuper,yuan2026super}. 
As a consequence of this extension, one obtains integrable systems involving parabosonic and parafermionic degrees of freedom. 
$\Z2$-graded extensions of sine-Gordon, Liouville, sinh-Gordon, KdV and modified KdV equations have been constructed and their properties have been analyzed in the literature. 
Furthermore, $\Z2$-graded generalizations of the Virasoro algebra associated with the $\Z2$-graded Liouville equation have been derived.  
These results demonstrate that $\Z2$-graded Lie (super)algebras offer an exciting new direction in the study of integrable systems. 

These developments naturally raise the question of whether other integrable systems can also be incorporated into the $\Z2$-graded framework. 
The purpose of the present work is to consider a $\Z2$-graded extension of the Camassa-Holm equation (CH equation for short) \cite{CHoriginal,FUCHSSTEINERFOKAS}, which is a prototypical example of an integrable nonlinear PDE.

The CH equation has been proposed as a model for shallow water waves, with its peaked soliton solutions, known as \textit{peakons}, describing the propagation of shallow water waves. 
It is one of the most extensively studied nonlinear PDEs and  exhibits several mutually related manifestations of integrability, including a Lax operator formulation, a bi-Hamiltonian structure \cite{CHoriginal,FUCHSSTEINERFOKAS}, B\"acklund transformation \cite{RasinSchiff}, an inverse scattering transform  \cite{CoGeIva}, a geometric interpretation as geodesic flow on the Bott-Virasoro group \cite{ChSchiff,MISIOLEK,Kouran,ConKole,Khesin,Kolev}, a conformal nature \cite{ivanov2005,ivanov2007conformal} and a description in terms of pseudo-spherical surfaces \cite{Reyes,HerReyes}. 

It has also been extensively generalized in various directions while preserving integrability. We mention here two such directions: super extensions \cite{ChSchiff,DevSch,POPOWICZ,LeneLech,ZhangZuo,GengXueWu,LiTan} and polynomial homogeneous generalizations, including the Novikov equation \cite{NovikovOrg}. 

In this paper, we apply the algebraic construction developed in \cite{AiFujiIto} (see also \cite{aratyn2004algebraic,Aratyn2001,Ferreira}) to derive a $\Z2$-graded extension of the CH equation. 
For the underlying algebraic structure, we employ the loop algebra of  the $\Z2$-graded extension of $\osp(1|2)$. 
We then focus on bi-Hamiltonian structure of the derived equation. 
A remarkable feature of the $\Z2$-graded CH equation is that it admits infinitely many conserved quantities of non-trivial $\Z2$-degree (see \S \ref{SEC:algebra} for the definition of the $\Z2$-degree). We also show that all the conserved quantities, of both trivial and non-trivial $\Z2$-degrees, are in involution with respect to the $\Z2$-graded Poisson bracket. These results demonstrate that the $\Z2$-grading gives rise to a rich integrable structure beyond that of the ordinary CH equation. 

This paper is organized as follows. 
In \S \ref{SEC:algebra}, we present the general definition of  $\Z2$-graded Lie superalgebras and introduce the $\Z2$-graded loop algebra of $\osp(1|2)$.  
In \S \ref{SEC:CHeq}, we use the loop algebra to construct integrable systems via the zero-curvature equation. This construction leads to a $\Z2$-graded extension of the CH equation, which we refer to as the $\Z2$-CH equation. 
Then, in \S \ref{SEC:SZCE}, we employ the stationary zero-curvature equation to derive a recurrence relation which is essential to establish the bi-Hamiltonian structure of the $\Z2$-CH equation. 
The bi-Hamiltonian structure is established in \S \ref{SEC:BHS}, where we introduce an infinite sequence of Hamiltonian functionals with both trivial and non-trivial $\Z2$-degrees and prove that they are in involution with respect to the $\Z2$-graded Poisson bracket. 
Finally, we summarize our results and comment on some open problems in \S \ref{SEC:CR}.

\section{Loop extension of $\Z2$-graded $\osp(1|2)$} \label{SEC:algebra}
\setcounter{equation}{0}

Let us recall the definition of the $\Z2$-graded Lie superalgebra \cite{rw1,rw2,scheu}. 
Let $\mathfrak{g}$ be a vector space and $\hat{a}\equiv [a_1a_2]$ an element of $\Z2$. 
Suppose that $\mathfrak{g}$ is a direct sum of graded components 
\begin{equation}
	\mathfrak{g} = \bigoplus_{\hat{a} \in \Z2}  \g_{\hat{a}}= \g_{[00]} \oplus \g_{[10]} \oplus \g_{[01]} \oplus \g_{[11]}.   
\end{equation}
If $\mathfrak{g}$ admits a bilinear operation (the graded Lie bracket), denoted by $ \llbracket \cdot, \cdot \rrbracket, $ and satisfying the identities 
\begin{align}
	& \llbracket A_{\hat{a}}, B_{\hat{b}}  \rrbracket \in \g_{\hat{a}+\hat{b}},
	\\
	& \llbracket A_{\hat{a}}, B_{\hat{b}} \rrbracket = -(-1)^{\hat{a}\cdot\hat{b}} \llbracket  B_{\hat{b}}, A_{\hat{a}} \rrbracket,
	\\
	& (-1)^{\hat{a}\cdot\hat{c}} \llbracket A_{\hat{a}}, \llbracket B_{\hat{b}}, C_{\hat{c}} \rrbracket \rrbracket + \PH{b}{a} \llbracket B_{\hat{b}}, \llbracket C_{\hat{c}}, A_{\hat{a}} \rrbracket \rrbracket + \PH{c}{b}\llbracket C_{\hat{c}}, \llbracket A_{\hat{a}}, B_{\hat{b}} \rrbracket \rrbracket  = 0, \label{AlgJacobi}
\end{align}
where $ A_{\hat{a}}, B_{\hat{a}}, C_{\hat{a}} $ are homogeneous elements of $\mathfrak{g}_{\hat{a}}$ and 
\begin{equation}
	\hat{a} + \hat{b} = [a_1+b_1 ,a_2+b_2] \in \Z2, \qquad 
	\hat{a} \cdot \hat{b} = a_1 b_1 + a_2 b_2 \in \mathbb{Z}_2,
\end{equation}
then $\mathfrak{g}$ is referred to as a $\Z2$-graded Lie superalgebra. 
We define the \textit{degree} of $ A_{\hat{a}} $ by the corresponding element of $\Z2$: $ \deg A_{\hat{a}} = \hat{a}. $

It is clear from the definition that the graded Lie brackets are realized by commutators and anticommutators as follows 
\begin{equation}
	\llbracket A_{\hat{a}}, B_{\hat{b}} \rrbracket
	= 
	\begin{cases}
		[A_{\hat{a}}, B_{\hat{b}}], & \DP{a}{b} = 0,
		\\[10pt]
		\{ A_{\hat{a}}, B_{\hat{b}} \}, & \DP{a}{b} = 1.
	\end{cases}
\end{equation}
If $\llbracket A, B \rrbracket = 0, $ we say that $A$ and $ B$ are $\Z2$-\textit{graded commutative}. 
It is also observed from the definition that $ \g$ has $\mathbb{Z}_2$-grading, too:
\begin{equation}
	\g = \g_{[0]} \oplus \g_{[1]}, \qquad \g_{[0]} := \g_{[00]} \oplus \g_{[11]}, \ 
	\g_{[1]} := \g_{[10]} \oplus \g_{[01]}.
\end{equation}

In the present work, we take the $\Z2$-graded extension of the Lie superalgebra $\osp(1|2)$ introduced in \cite{Aizawa_2020,AmaAi} and consider its loop extension. 
This infinite-dimensional $\Z2$-graded loop superalgebra is also denoted by $\g$; this should not cause any confusion. Each $\Z2$-graded subspace of $\g$ is also infinite-dimensional and is given by
\begin{equation}
	\g_{[00]} = \{ \ K^0_m, \ K^{\pm}_m \ \}, 
	\quad
	\g_{[10]} = \{  P^{\pm}_m \ \}, 
	\quad  
	\g_{[01]} = \{  Q^{\pm}_m \ \}, 
	\quad
	\g_{[11]} = \{ \ L^0_m, \ L^{\pm}_m \ \}
\end{equation}
where $m \in \mathbb{Z}.$ 
We now present their non-vanishing relations in terms of (anti)commutators. 
The $\g_{[0]}$ sector forms a Lie subalgebra	
\begin{alignat}{3}
	[K^0_m, K^{\pm}_n] &= \pm 2 K^{\pm}_{m+n}, & \qquad 
	[K^+_m, K^-_n] &= K^0_{m+n}, & \qquad 
	[L^0_m, L^{\pm}_n] &= \pm 2 K^{\pm}_{m+n},
	\nn \\
	[L^+_m, L^-_n] &= K^0_{m+n}, & 
	[K^0_m, L^{\pm}_n] &= \pm 2 L^{\pm}_{m+n}, &
	[K^{\pm}_m, L^{\mp}_n] &= \pm L^0_{m+n},
	\nn \\
	[L^0_m, K^{\pm}_n] &= \pm 2 L^{\pm}_{m+n}.
\end{alignat}
The relations between $ \g_{[0]}$ and $ \g_{[1]}$ are determined by commutators and anticommutators:
\begin{alignat}{3}
	[K^0_m, P^{\pm}_n] &= \pm P^{\pm}_{m+n}, &\qquad    
	[K^{\pm}_m, P^{\mp}_n] &= -P^{\pm}_{m+n}, &\qquad
	[K^0_m, Q^{\pm}_n] &= \pm Q^{\pm}_{m+n},
	\nn \\[3pt]
	[K^{\pm}_m, Q^{\mp}_n] &= -Q^{\pm}_{m+n}, & 
	\{ L^0_n, P^{\pm}_m  \} &= \pm i Q^{\pm}_{m+n}, & 
	\{  L^{\mp}_n, P^{\pm}_m \} &= -i Q^{\mp}_{m+n},
	\nn \\[3pt]
	\{ L^0_n, Q^{\pm}_m \} &= \mp i P^{\pm}_{m+n},  & 
	\{ L^{\mp}_n, Q^{\pm}_m \} &= i P^{\mp}_{m+n}
\end{alignat}
and the same applies to the $\g_{[1]}$-sector 
\begin{alignat}{3}
	\{ P^{\pm}_m, P^{\pm}_n \} &= \pm 2 K^{\pm}_{m+n}, &\qquad 
	\{P^+_m, P^-_n \} &= K^0_{m+n}, & \qquad 
	[P^{\pm}_m, Q^{\pm}_n] &= \pm 2i L^{\pm}_{m+n},
	\nn \\[3pt]
	[P^{\pm}_m, Q^{\mp}_n] &= i L^0_{m+n}, &
	\{ Q^{\pm}_m, Q^{\pm}_n\} &= \pm 2 K^{\pm}_{m+n}, & \{Q^+_m, Q^-_n\} &=K^0_{m+n}.
\end{alignat}
The elements with $m = 0 $ correspond to the ten-dimensional $\Z2$-graded Lie superalgebra \cite{Aizawa_2020,AmaAi}
\begin{equation}
	\Z2\text{-}\osp(1|2) = \{ \ K^0_0, \ K^{\pm}_0,\ P^{\pm}_0, \ Q^{\pm}_0, \ L^0_0, \ L^{\pm}_0 \  \}. 
	\label{FiniteDalg}
\end{equation}

One may introduce $\Z2$-graded derivation $ d_{\hat{m}}$ for $ \hat{m} \in \Z2$ by  
\begin{equation}
	d_{\hat{m}}(\llbracket A_k, B_n \rrbracket) = 
	\llbracket d_{\hat{m}}(A_k), B_n \rrbracket + \PH{m}{a} \llbracket A_k, d_{\hat{m}}(B_n) \rrbracket, 
	\quad A_n \in \g_a, \ B_k \in \g_b
\end{equation}
It was shown in \cite{AiSe} that $\g$ admits the [00] and [11]-graded derivations, but none of [10], [01]-graded. 
The [00]-graded derivation counts the mode as usual {\color{mycol}(Recall that $ X_m := \lambda^m \otimes X$ where $ X \in \Z2\text{-}\osp(1|2)$)}:
\begin{equation}
	[d_{00}, X_m] = mX_m, \quad X_m \in \g
\end{equation}
The action of [11]-graded derivation is determined by
\begin{alignat}{3}
	[d_{11}, K_m^0] &= m L^0_m, & \qquad [d_{11}, L_m^0] &= m K_m^0, & \qquad [d_{11}, K_m^{\pm}] &= m L^{\pm}_m,
	\notag \\[3pt]
	[d_{11}, L_m^{\pm}] &= m K_m^{\pm}, & \{d_{11}, P_m^{\pm}\} &= i m Q_m^{\pm},  & \{ d_{11}, Q_m^{\pm}\} &= -im P_m^{\pm},
	\notag \\[3pt]
	[d_{11}, d_{00}] & = 0.
\end{alignat}

The $\Z2$-loop superalgebra $\g$ admits various gradations. 
The typical example, which is used in this paper, is the homogeneous gradation defined by using the derivation $ d_{00}$:
\begin{align}
	\g = \bigoplus_{n \in \mathbb{Z}} \g_n, \qquad 
	\g_m = \{\  X \in \g \ \left| \ [d_{00}, X] = m X \right. \ \}. \label{HomoG}
\end{align}

\section{Lax operator formulation of $\Z2$-graded CH equation} \label{SEC:CHeq}

Let us consider the following Lax operators valued in $\g$:
\begin{align}
	L_x &=K_0^++\frac{1}{4}K_0^-+\alpha P^-_1+\beta Q^-_1+mK_2^-+vL_2^- \ \in \g_0 \oplus \g_1 \oplus \g_2,
	\notag \\
	L_t &= \sum_{k}\big( \mathsf{D}^{2k} + \mathsf{D}^{2k+1} \big), 
	\label{LaxOp}
\end{align}
where
\begin{align}
	\mathsf{D}^{2k}&:= \mathsf{a}_{2k}K^+_{2k}+ \mathsf{b}_{2k}K^{-}_{2k}+ \mathsf{c}_{2k}K^0_{2k}+ \mathsf{f}_{2k}L^+_{2k}+ \mathsf{g}_{2k}L^{-}_{2k}+ \mathsf{h}_{2k}L^0_{2k} \ \in \g_{2k},\\
	\mathsf{D}^{2k+1}&:=\upxi_{2k+1}P^+ _{2k+1}+\upeta_{2k+1}P^- _{2k+1}+\uprho_{2k+1}Q^+_{2k+1}+\upsigma_{2k+1}Q^-_{2k+1} \ \in \g_{2k+1}.
\end{align}
The coefficients of the basis elements of $\g$, such as $\alpha$ and $\beta$, are functions on $(1+1)$-dimensional spacetime taking values in a $\Z2$-graded commutative algebra. Each coefficient has the same degree as the corresponding basis element of $\g$, so that the Lax operators $L_x$ and $L_t$ are $[00]$-graded. 
In particular, the coefficients in $L_x$ will play the role of dynamical variables of our theory.  We therefore specify their degrees explicitly:
\[
   \alpha(t,x) \quad [10], \qquad \beta(t,x) \quad  [01], \qquad m(t,x) \quad [00], \qquad v(t,x) \quad [11].
\]
The $\Z2$-graded commutativity implies that $m(t,x)$ commutes with all other variables and 
\begin{align}
	\alpha^2 = \beta^2 = [\alpha, \beta] = \{\alpha,v\} = \{\beta, v\} = 0.
		\label{RelationDV}
\end{align}

 Once $L_t$ is fixed, namely, once the range of $k$ is specified, the zero-curvature equation determines the coefficients in $L_t$ in terms of those in $L_x$, thereby yielding an evolution equation. 
To obtain a $\Z2$-graded extension of the CH equation, we choose $L_t$ as follows:
\begin{equation}
	L_t = \mathsf{D}^{-2} + \mathsf{D}^{-1} + \mathsf{D}^0 + \mathsf{D}^1 + \mathsf{D}^2.
\end{equation}
Then, imposing the zero-curvature equation
\begin{equation}
	\partial_t L_x - \partial_x L_t + [L_x, L_t] = 0.
\end{equation}
we can determine the nonvanishing coefficients: 
\begin{alignat}{3}
	\mathsf{a}_{-2}&=\frac{1}{2}, & \mathsf{a}_0 &= -u,
	\notag \\
	\mathsf{b}_{-2} &= \frac{1}{8}, & \mathsf{b}_0 &= \frac{1}{4}(u+\psi\psi'+\chi\chi'), &\qquad  \mathsf{b}_2 &=-mu-vw,
	\notag\\
	\mathsf{c}_0 &=\frac{1}{2}u', & \mathsf{f}_0 &= -w, & \mathsf{g}_0 &=\frac{1}{4}(w-i\psi\chi'+i\psi'\chi),
	\notag \\
	\mathsf{g}_2 &=-mw-vu, & \mathsf{h}_0 &=\frac{1}{2}w',
	\notag \\
	\upxi_{-1} &=-\frac{1}{2}\psi', & \uprho_{-1} &=-\frac{1}{2}\chi', & \upeta_{-1} &=\frac{1}{8}\psi,
	\notag \\
	\upeta_1 &=-u\alpha-iw\beta, &\qquad \upsigma_{-1} &=\frac{1}{8}\chi, & \upsigma_1 &=-u\beta+iw\alpha,
\end{alignat}
where we introduced new $\Z2$-graded commutative variables
\begin{equation}
	u(t,x) \quad [00], \qquad w(t,x) \quad [11], \qquad \psi(t,x) \quad [10], \qquad \chi(t,x) \quad [01].
\end{equation}
Namely, $u(t,x)$ commutes with all other variables and
\begin{equation}
	\{w, \psi\} = \{w,\chi\} = [\psi, \chi] = \psi^2 = \chi^2 = 0.
\end{equation}
With these nonvanishing coefficients, the evolution equations derived from the zero-curvature equation are given by
\begin{align}
	\left\{
	\begin{aligned}
		\dot{m}+m'u+2mu'+v'w+2vw'&=0,\\
		\dot{v}+m'w+2mw'+v'u+2vu'&=0,\\
		\dot{\alpha}+\alpha'u+\frac{3}{2}\alpha u'
		+iw\beta'+\frac{3i}{2}w'\beta
		+\frac{1}{2}m\psi'+\frac{i}{2}v\chi'&=0,\\
		\dot{\beta}+\beta'u+\frac{3}{2}\beta u'
		-iw\alpha'-\frac{3i}{2}w'\alpha
		+\frac{1}{2}m\chi'-\frac{i}{2}v\psi'&=0.
	\end{aligned}
	\right.\label{z2z2ch}
\end{align}
where the relation between $(m,v,\alpha,\beta)$ and $(u,w,\psi,\chi)$ is given by
\begin{align}
	m &=u-u''+\frac{3}{4}(\psi\psi'+\chi\chi')+(\psi'\psi''+\chi'\chi''),
	\notag\\
	v &=w-w''+\frac{3i}{4}(\psi'\chi-\psi\chi')+i(\psi''\chi'-\psi'\chi''),
	\notag\\
	\alpha &=\frac{1}{4}\psi-\psi'',
	\notag \\
	\beta &=\frac{1}{4}\chi-\chi'' \label{momentum}.
\end{align}
The dot and the prime denote time and space derivatives, respectively; for example,  $\dot{m} = \partial_t m, u' = \partial_x u.$ 
As can be seen from \eqref{momentum}, the variables $m, v, \alpha$ and $ \beta $ form a $\Z2$-graded generalization of the momentum density of the CH equation. 
We refer to  this system of equations as the \textit{$\Z2$-CH equation}. 
The equation \eqref{z2z2ch} admits the following expression:
\begin{equation}
	\begin{pmatrix}
		\dot{m} \\ \dot{v} \\ \dot{\alpha} \\ \dot{\beta}
	\end{pmatrix}  
	=-
	\begin{pmatrix}
		\partial_x m + m\partial_x & \partial_x v + v \partial_x & 0 & 0
		\\[3pt]
		\partial_x v + v \partial_x & \partial_x m + m\partial_x & 0 & 0
		\\[3pt]
		\partial_x \alpha + \frac{1}{2}\alpha \partial_x & - i(\partial_x \beta + \frac{1}{2}\beta \partial_x) & \frac{1}{2}m & \frac{i}{2}v
		\\[3pt]
		\partial_x \beta + \frac{1}{2}\beta \partial_x & i(\partial_x \alpha + \frac{1}{2}\alpha \partial_x) & -\frac{i}{2}v & \frac{1}{2}m
	\end{pmatrix}
	\begin{pmatrix}
		u \\ w \\ \psi' \\ \chi'
	\end{pmatrix} 
	\equiv J_1 
	\begin{pmatrix}
		u \\ w \\ \psi' \\ \chi'
	\end{pmatrix}. \label{matrixch}
\end{equation}

Setting all the nontrivially graded variables, namely those of degrees $[11]$, $[10]$, and $[01]$, to zero, the $\Z2$-CH equation reduces to the ordinary CH equation. 
Setting the variables of degrees [11] and [10] (or [11] and [01]) to zero, the $\Z2$-CH equation reduces to the super CH equation introduced in \cite{GengXueWu}.  

Although we focused only on the $\Z2$-graded extension of the CH equation, extending the range of $L_t$ would yield, via the zero-curvature equation, a hierarchy containing the $\Z2$-CH equation, similar to the CH hierarchy discussed in \cite{Qiao}.

%
\section{Stationary zero-curvature equation} \label{SEC:SZCE}

In the next section, we present the bi-Hamiltonian structure of the $\Z2$-CH equation. 
As developed in \cite{Tu1986,Tu1989,Tu1989trace}, the stationary zero-curvature equation provides a powerful tool for constructing bi-Hamiltonian structures of integrable systems. Its super extension ($\mathbb{Z}_2$-graded extension) has been developed and applied to various integrable systems \cite{SupeTraceId,SuperKdv,GengXueWu,LiTan,SupeDP}. 
In this section, we extend this construction to the $\Z2$-graded setting.

For the Lax operator $L_x$ given in \eqref{LaxOp}, the stationary zero-curvature equation is given by
\begin{equation}
	\partial_xW=[L_x,W], \label{SZCE}
\end{equation}
where  $W \in \g $ has the same structure as $L_t$:
\begin{align}
	W &=\sum_{k\in \mathbb Z}(D^{-2k}+D^{-2k-1}),
	\notag \\
	D^{-2k} &=a_{-2k}K^+_{-2k}+b_{-2k}K^-_{-2k}+c_{-2k}K^0_{-2k}+f_{-2k}L^+_{-2k}+g_{-2k}L^-_{-2k}+h_{-2k}L^0_{-2k},
	\notag \\
	D^{-2k-1}&=\xi_{-2k-1}P^+_{-2k-1}+\eta_{-2k-1}P^-_{-2k-1}+\rho_{-2k-1}Q^+_{-2k-1}+\sigma_{-2k-1}Q^-_{-2k-1}. 
	\label{Wdef}
\end{align}
Here the coefficients of the basis elements of $\g$ are functions $(t,x)$ taking values in a $\Z2$-graded commutative algebra.

The stationary zero-curvature equation \eqref{SZCE} gives the following relations for the $\Z2$-graded commutative functions:
\begin{align}
	a'_{-2n} &=-2c_{-2n},
	\notag\\
	b'_{-2n} &=\frac{1}{2}c_{-2n}-2(\eta_{-2n-1}\alpha +\sigma_{-2n-1}\beta)+2(mc_{-2n-2}+vh_{-2n-2}),
	\notag\\
	c'_{-2n} &=b_{-2n}-\frac{1}{4}a_{-2n}+\xi_{-2n-1}\alpha+\rho_{-2n-1}\beta -ma_{-2n-2}-vf_{-2n-2},
	\notag\\
	f'_{-2n} &=-2h_{-2n},
	\notag\\
	g'_{-2n} &=\frac{1}{2}h_{-2n}+2i(\eta_{-2n-1}\beta -\sigma_{-2n-1}\alpha)+2(mh_{-2n-2}+vc_{-2n-2}),
	\notag\\
	h'_{-2n} &=g_{-2n}-\frac{1}{4}f_{-2n}+i(\rho_{-2n-1}\alpha-\xi_{-2n-1}\beta)-mf_{-2n-2}-va_{-2n-2},
	\notag\\
	\xi'_{-2n-1} &=-\eta_{-2n-1}+a_{-2n-2}\alpha+if_{-2n-2}\beta,
	\notag\\
	\eta'_{-2n-1} &=-\frac{1}{4}\xi_{-2n-1}+c_{-2n-2}\alpha+ih_{-2n-2}\beta+i\rho_{-2n-3}v-\xi_{-2n-3}m,
	\notag\\
	\rho_{-2n-1}'&=-\sigma_{-2n-1}+a_{-2n-2}\beta -if_{-2n-2}\alpha,
	\notag\\
	\sigma'_{-2n-1} &=-\frac{1}{4}\rho_{-2n-1}+c_{-2n-2}\beta-ih_{-2n-2}\alpha-\rho_{-2n-3}m-i\xi_{-2n-3}v.
\end{align}
These relations reduce to the recurrence relations  for $a_{-2k},f_{-2k},\xi_{-2k-1}$ and $ \rho_{-2k-1}$:
\begin{align}
	&\frac12 K a_{-2n}
	+A\xi_{-2n-1}
	+B\rho_{-2n-1}
	=
	M a_{-2n-2}
	+ V f_{-2n-2},
	\nonumber\\[1ex]
	&\frac12 K f_{-2n}
	+iB\xi_{-2n-1}
	-iA\rho_{-2n-1}
	=
	Va_{-2n-2} + M f_{-2n-2},
	\nonumber\\[1ex]
	&L \xi_{-2n-1}
	= \frac{1}{2}A^{\dagger} a_{-2n-2}
	-\frac{i}{2} B^{\dagger} f_{-2n-2}
	-m\,\xi_{-2n-3}
	-i v\,\rho_{-2n-3},
	\nonumber\\[1ex]
	&L\rho_{-2n-1}
	=
	\frac{1}{2} B^{\dagger} a_{-2n-2}
	+ \frac{i}{2}A^{\dagger} f_{-2n-2}
	+i v\,\xi_{-2n-3}
	-m\,\rho_{-2n-3},
	\label{SZCrecrel}
\end{align}
where the $\Z2$-graded operators are defined by
\begin{alignat}{3}
	K &=  \partial_x - \partial_x^3, & \qquad L &= \frac{1}{4}-\partial_x^2, & \qquad M &=-(m\partial_x + \partial_x m),
	\notag \\
	V &= -(v\partial_x + \partial_x v), & A &= 2\alpha \partial_x + \partial_x \alpha, & 
	B&= 2\beta \partial_x + \partial_x \beta,
	\label{KLMAB}
\end{alignat}
and their adjoints are determined by the following definition:
%
\begin{dfn}
The adjoint ${\mathcal O}^{\dagger}$ of a $\Z2$-graded operator  $\mathcal{O}$ is defined by
\begin{equation}
	\int f \mathcal{O} g dx = (-1)^{\hat{f}\cdot \hat{\mathcal{O}}} \int (\mathcal{O}^{\dagger}f) g dx \label{DEF:ajoint}
\end{equation}
for all $\Z2$-graded functions $ f$ and $g $ of homogeneous degree  $ \hat{f}$ and $ \hat{g}.$ 
\end{dfn}
Note that, unlike the usual adjoint operation,  the definition of the adjoint does not involve complex conjugation. 
It follows from this definition that
\begin{align}
	(\mathcal{O}_1 \mathcal{O}_2)^\dagger &= (-1)^{\hat{\mathcal{O}}_1 \cdot \hat{\mathcal{O}}_2}(\mathcal{O}_2^\dagger \mathcal{O}_1^\dagger) 
	\notag 
	\\
	(\mathcal{O}^\dagger)^\dagger &=\mathcal{O}. \label{daggero}
\end{align}
It is also straightforward to verify that the adjoints of the operators in \eqref{KLMAB} are given by
\begin{alignat}{2}
	K^{\dagger} &= -K, \qquad L^{\dagger} = L, &\quad M^{\dagger} &= -M, \quad V^{\dagger} =-V,
	\notag\\
	A^{\dagger} &= -(\alpha\partial_x + 2\partial_x \alpha), &
	B^{\dagger} &= -(\beta\partial_x + 2\partial_x \beta).
\end{alignat}

The recurrence relation \eqref{SZCrecrel} can be written in terms of  the matrix differential operator $J_1$ defined in \eqref{matrixch} and $J_2$ defined below: 
\begin{equation}
	J_1
	\begin{pmatrix}
		a_{-2n-2}\\
		f_{-2n-2}\\
		2\xi_{-2n-3}\\
		2\rho_{-2n-3}
	\end{pmatrix}
	=
	J_2
	\begin{pmatrix}
		a_{-2n}\\
		f_{-2n}\\
		2\xi_{-2n-1}\\
		2\rho_{-2n-1}
	\end{pmatrix} \label{rec_rel}
\end{equation}
where
\begin{align}
	J_1&=
	\begin{pmatrix}
		M & V & 0 & 0
		\\[2mm]
		V & M & 0 & 0
		\\[2mm]
		\frac{1}{2}A^\dagger & -\frac{i}{2}B^\dagger
		& -\frac{1}{2}m
		& -\frac{i}{2}v
		\\[2mm]
		\frac{1}{2}B^\dagger & \frac{i}{2}A^\dagger
		& \frac{i}{2}v
		& -\frac{1}{2}m
	\end{pmatrix},
	\qquad 
	J_2=\frac{1}{2}
	\begin{pmatrix}
		K & 0 & A & B
		\\[2mm]
		0 & K & iB & -iA
		\\[2mm]
		0 & 0 & L & 0
		\\[2mm]
		0 & 0 & 0 & L
	\end{pmatrix}.
\end{align}
One may verify that one of the solutions of the recurrence relation \eqref{rec_rel} is given by
\begin{equation}
	J_1
	\begin{pmatrix}
		\frac12\\
		0\\
		0\\
		0
	\end{pmatrix}
	=
	J_2
	\begin{pmatrix}
		-u\\
		-w\\
		-\psi'\\
		-\chi'
	\end{pmatrix}.
\end{equation}

We further rewrite the recurrence relation \eqref{rec_rel} using the following relation: 
\begin{align}
	\begin{pmatrix}
		a_{-2n} \\ f_{-2n} \\ 2\xi_{-2n-1} \\ 2\rho_{-2n-1}
	\end{pmatrix}
	&=
	\begin{pmatrix}
		2K^{-1}M & 2K^{-1}V & -K^{-1}A & -K^{-1}B
		\\
		2K^{-1}V & 2K^{-1}M & -iK^{-1}B & iK^{-1}A
		\\
		0 & 0 & 1 & 0
		\\
		0 & 0 & 0 & 1
	\end{pmatrix}
	\begin{pmatrix}
		a_{-2n-2} \\ f_{-2n-2} \\ 2\xi_{-2n-1} \\ 2\rho_{-2n-1}
	\end{pmatrix}
	\notag \\
	&\equiv S    
	\begin{pmatrix}
		a_{-2n-2} \\ f_{-2n-2} \\ 2\xi_{-2n-1} \\ 2\rho_{-2n-1}
	\end{pmatrix} \label{mat_S} .
\end{align}
This relation follows from the first two relations in \eqref{SZCrecrel}. 
Substituting \eqref{mat_S} into \eqref{rec_rel}, we obtain the new recurrence relation
\begin{align}
	P_1 \begin{pmatrix}
		a_{-2n-4}\\
		f_{-2n-4}\\
		2\xi_{-2n-3}\\
		2\rho_{-2n-3}
	\end{pmatrix}
	=
	P_2 \begin{pmatrix}
		a_{-2n-2}\\
		f_{-2n-2}\\
		2\xi_{-2n-1}\\
		2\rho_{-2n-1}
	\end{pmatrix}. \label{recurrence}
\end{align}
Here the new matrix differential operators are defined by
\begin{align}
	P_1&:= J_1S=
	\begin{pmatrix}
		p_{1} &p_{2} &p_{3} &p_{4}\\
		p_{2} &p_{1} &ip_{4} &-ip_{3}\\
		p_{5} &p_{6} &p_{7} &p_{8}\\
		ip_{6} &ip_{5} &-p_{8} &p_{7}
	\end{pmatrix},
	\notag \\[3pt]
	P_2&:= J_2 S= 
	\begin{pmatrix}
		M&V&0&0\\
		V&M&0&0\\
		0&0&\frac12 L&0\\
		0&0&0&\frac12 L
	\end{pmatrix}, \label{P1P2def}
\end{align}    
where
\begin{align}
	p_{1}&= 2MK^{-1}M+2VK^{-1}V, &
	p_{2}&= 2VK^{-1}M+2MK^{-1}V, \notag\\
	p_{3}&= -MK^{-1}A-iVK^{-1}B, &
	p_{4}&= -MK^{-1}B +iVK^{-1}A,\notag\\
	p_{5}&= A^\dag K^{-1}M -iB^\dag K^{-1}V, & 
	p_{6}&= A^\dag K^{-1}V -iB^\dag K^{-1}M,\notag\\ 
	p_{7}&= -\frac12 A^\dag K^{-1}A -\frac12 B^\dag K^{-1}B-\frac12 m, & 
	p_{8}&= -\frac12 A^\dag K^{-1}B +\frac12 B^\dag K^{-1} A -\frac{i}{2}v.
\end{align}

In terms of the operators $P_1$ and $ P_2$, the $\Z2$-CH equation takes the form
\begin{align}
	\begin{pmatrix}
		\dot{m}\\ \dot{v} \\ \dot{\alpha} \\ \dot{\beta}
	\end{pmatrix}
	&=P_1
	\begin{pmatrix}
		-\frac12\\0\\\psi'\\\chi'
	\end{pmatrix}
	\notag \\
	&=P_2
	\begin{pmatrix}
		u\\
		w\\
		-2L^{-1}\left(\alpha' u +\frac32 \alpha u' +iw\beta' +\frac{3i}{2} w'\beta + \frac{1}{2}m\psi' +\frac{i}{2}v\chi'\right)\\
		-2L^{-1}\left(\beta' u +\frac32 \beta u' -iw\alpha' -\frac{3i}2 w'\alpha + \frac{1}{2}m\chi' -\frac{i}{2}v\psi'\right)
	\end{pmatrix}. \label{EOM_P}
\end{align}

The main consequences of this section are the recurrence relation \eqref{recurrence} and the $\Z2$-CH equation in the form given by \eqref{EOM_P}. 
We will show in the next section that the operators $P_1 $ and $P_2$ define Poisson structures.

\section{Bi-Hamiltonian structure of $\Z2$-CH equation} \label{SEC:BHS}

In this section, we present the bi-Hamiltonian structure of the $\Z2$-CH equation  and use it to demonstrate the existence of infinitely many conserved quantities of degrees $[00]$ and $[11]$. 
We also show that all these conserved quantities are in involution. 
Our strategy is to generalize the theory developed by Olver \cite{Olver1993applications} to the $\Z2$-graded framework. The corresponding generalization to the $\mathbb{Z}_2$-graded (super) case was already developed by Mathieu \cite{Mathieu1988}. 
We use the same terminology as in \cite{Olver1993applications, Mathieu1988} for the $\Z2$-graded generalization. However, we modify the notation to suit the $\Z2$-graded setting.

\subsection{Hamiltonian functionals}

We introduce the following notation for the momentum densities:
\begin{align}
	\mathsf{m} = (m_1,m_2,m_3,m_4)^T :=(m,v,\alpha,\beta)^T.
\end{align}
We then consider the following $\Z2$-graded jet space with coordinates
\begin{eqnarray}
	(t,x,m_i, m_i^{(n)} ), \quad m_i^{(n)}:=\partial_x^n\, m_i, \quad n = 1,2,\dots
\end{eqnarray}
where $i = 1, \dots, 4.$ Here, $ m_i = m_i^{(0)}$ is understood. 
When no confusion arises, we denote the degree of $m_i^{(n)}$ simply by $\hat{i}.$ 
The algebra generated by the coordinates of this jet space forms a $\Z2$-graded commutative algebra, which naturally generalizes the coordinate algebra of the jet space used in the study of PDEs.

Now we define the left derivative by
\begin{align}
	\overrightarrow{\frac{\partial}{\partial m_i^{(n)}}}\, m_j^{(k)} = \delta_{ij}\, \delta^{nk},
\end{align}
and for homogeneous functions $ F, G$ on the jet space
\begin{eqnarray}
	\overrightarrow{\frac{\partial}{\partial m_i^{(n)}}}\, FG = \left(\overrightarrow{\frac{\partial}{\partial m_i^{(n)}}}F \right) G 
	+ (-1)^{\hat{i}\cdot \hat{F}} F  \overrightarrow{\frac{\partial}{\partial m_i^{(n)}}}\, G.
\end{eqnarray}
The right derivative is defined similarly:
\begin{align}
	m_j^{(k)} 
	\overleftarrow{\frac{\partial}{\partial m_i^{(n)}}}  &= \delta_{ij} \,\delta^{nk},
	\notag \\
	FG \overleftarrow{\frac{\partial}{\partial m_i^{(n)}}} &= 
	F \left( G \overleftarrow{\frac{\partial}{\partial m_i^{(n)}}} \right)  
	+ (-1)^{\hat{G}\cdot \hat{i}} \left(F \overleftarrow{\frac{\partial}{\partial m_i^{(n)}}}\right) G.
\end{align}

\begin{dfn}
	We define the left and right Euler operators with respect to $m_i$, denoted by $L_i$ and $R_i$, respectively, as
\begin{align}
	L_i=\sum_{n\ge 0} \overrightarrow{(-\partial_x)^n\frac{\partial}{\partial m_i^{(n)}}},\qquad
	R_i=\sum_{n\ge 0} \overleftarrow{\frac{\partial}{\partial m_i^{(n)}}(-\partial_x)^n}. 
	\label{DefEuler}
\end{align}	
\end{dfn}

\begin{prop} \label{Prop:Euler}
	For a homogeneous function $F$ on the jet space, we have 
\begin{align}
	R_i (F)=(-1)^{(\hat i+\hat F)\cdot \hat i} L_i(F) \label{lemmaeuler}
\end{align}
\end{prop}
\begin{proof}
	It is sufficient to show that
	\begin{eqnarray}
		F\, \overleftarrow{\frac{\partial}{\partial m_i^{(n)}}} = 
		(-1)^{(\hat i + \hat F)\cdot \hat i} \overrightarrow{\frac{\partial}{\partial m_i^{(n)}}}\, F. \label{RLcomp}
	\end{eqnarray}
	Suppose that the right derivatives of $F$ do not vanish. By bringing all $m_i^{(n)}$ to the right,  $F$ can be written as
	\begin{equation}
		F = G (m_i^{(n)})^k, \qquad \deg G = \hat F + k\hat i
	\end{equation}
	where  $k$ is a positive integer and a function $G$ is independent of $ m_i^{(n)}.$ 
	It follows that
	\begin{eqnarray}
		F\, \overleftarrow{\frac{\partial}{\partial mi^{(n)}}} = k G (m_i^{(n)})^{k-1}.
	\end{eqnarray}
	By the definition of the left derivative, we have
	\begin{align}
		\overrightarrow{\frac{\partial}{\partial m_i^{(n)}}}\, F 
		= 
		(-1)^{\hat i \cdot \hat G} k G  (m_i^{(n)})^{k-1}.
	\end{align}
	Since $ k\hat i \cdot \hat i = \hat i \cdot \hat i \mod(2,2),$ we have $\hat i \cdot \hat G = \hat i \cdot (\hat F + \hat i) \mod(2,2),$ which proves \eqref{RLcomp}.
\end{proof}

Let $ a(x,\mathsf{m},\mathsf{m}',\dots)$ be a homogeneous function depending on $x, \, \mathsf{m}$ and its derivatives on $x$ up to finite order on the jet space. Then 
\begin{eqnarray}
	A[\mathsf{m}] = \int dx\, a (x,\mathsf{m},\mathsf{m}',\dots)
\end{eqnarray}
is a functional having the same degree as $a$. 
\begin{dfn} \label{Def:FuncDiff}
	We define the left and right functional derivatives of $A$ with respect to $m_i$ in terms of the left and right Euler operator, respectively, by
	\begin{equation}
		\frac{\delta^L A}{\delta m_i}=L_i(a),\qquad
		\frac{\delta^R A}{\delta m_i}=R_i(a)
		.
	\end{equation}
\end{dfn}
The following relation  follows immediately from Proposition \ref{Prop:Euler}:
\begin{equation}
	\frac{\delta^R A}{\delta m_i} = (-1)^{(\hat i + \hat A)\cdot \hat i} \,\frac{\delta^L A}{\delta m_i}. 
	\label{RLfunctionalder}
\end{equation}

To introduce Hamiltonian functionals, we now turn to a matrix representation of the operators $L_x$ in \eqref{LaxOp} and $ W $ in \eqref{Wdef}.  
Introducing the loop parameter $\lambda$, we obtain a realization of $\g$:
\[
  K^{\pm}_{-2k} = K^{\pm}_0 \otimes \lambda^{-2k}, \quad P^{\pm}_{-2k-1} = P^{\pm}_0 \otimes \lambda^{-2k-1}, \;\dots
\]
Recall that $\{ \ K^0_0, \ K^{\pm}_0,\ P^{\pm}_0, \ Q^{\pm}_0, \ L^0_0, \ L^{\pm}_0 \  \}$ forms a basis of the ten-dimensional algebra $\mathfrak{k} := \Z2$-$\osp(1|2)$; see \eqref{FiniteDalg}. 
In terms of this basis, $W$ takes the form
\begin{align}
	W &= aK^+ + bK^- + cK^0 + f L^+ + g L^- + h L^0
	\notag \\
	  & + \xi P^+ + \eta P^- + \rho Q^+ + \sigma Q^-
\end{align}
where 
\begin{align}
	a := \sum_k a_{-2k} \lambda^{-2k}, \qquad 
	b := \sum_k b_{-2k} \lambda^{-2k}, \quad
	\xi := \sum_k \xi_{-2k-1} \lambda^{-2k-1}, \ \mathrm{etc.}
\end{align}
Here the subscript "0" of the basis of $\mathfrak{k}$ is omitted for simplicity of notation.

The algebra $\mathfrak{k}$ has a six-dimensional representation \cite{AFITTsuper}. 
First, observe that $\mathfrak{k}$ admits the triangular decomposition:
\begin{align}
	\mathfrak{k} &= \mathfrak{k}_+ \oplus \mathfrak{k}_0 \oplus \mathfrak{k}_-,
	\notag \\
	\mathfrak{k}_+ &= \mathrm{lin. sp.}\langle\; K^+, L^+, P^+, Q^+ \; \rangle,
	\notag \\
	\mathfrak{k}_0 &= \mathrm{lin. sp.}\langle\; K^0, L^0 \; \rangle,
	\notag \\
	\mathfrak{k}_- &= \mathrm{lin. sp.}\langle\; K^-, L^-, P^-, Q^- \; \rangle.
\end{align}
Based on this decomposition, we define the [00] and [11]-graded lowest weight vectors, denoted by $\ket{00}, \ket{11}$, respectively, by
\begin{align}
	P^-_0 \ket{00} &= Q^-_0 \ket{00} = P^-_0 \ket{11} = Q^-_0\ket{11} = 0,
	\notag \\[3pt]
	K^0_0 \ket{00} &= -\ket{00}, \quad K^0_0\ket{11} = -\ket{11},
	\quad
	L^0_0 \ket{00} = -\ket{11}, \quad L^0_0\ket{11} = -\ket{00}. \label{LWvec}
\end{align}
These vectors span a two-dimensional irreducible representation of the Borel subalgebra of  $\mathfrak{k}.$ 

We then consider a six-dimensional $\Z2$-graded vector space induced from this representation. Its basis elements and degrees are given by
\begin{alignat}{3}
	&[00] & \qquad \ket{1} &:= K^+ \ket{00}, &\qquad \ket{2} &:= \ket{00}, 
	\notag \\
	&[11] & \qquad \ket{3} &:= L^+ \ket{00}, &\ket{4} &:= \ket{11}, 
	\notag \\ 
	&[10] & \quad \ket{5} &:= P^+ \ket{00}, 
	\notag \\
	&[01] & \quad \ket{6} &:= Q^+ \ket{00}. \label{6DrepBasis}
\end{alignat}
Let $ \vec{v} = (\; \ket{1}, \ket{2}, \ket{3}, \ket{4}, \ket{5}, \ket{6}  \;) $ and 
define the action of $\mathfrak{k}$ on this basis by
\begin{align}
	\vec{v} & \stackrel{K^0}{\longrightarrow} (\; \ket{1}, -\ket{2}, \ket{3},-\ket{4},0, 0 \;),
	\notag \\
	\vec{v} & \stackrel{K^+}{\longrightarrow} (\; 0, \ket{1}, 0, \ket{3}, 0, 0 \;),
	\notag \\
	\vec{v} & \stackrel{K^-}{\longrightarrow} (\; \ket{2}, 0, \ket{4}, 0, 0, 0 \;),
	\notag \\
	\vec{v} & \stackrel{L^0}{\longrightarrow} (\; \ket{3}, -\ket{4}, \ket{1}, -\ket{2}, 0, 0 \;),
	\notag \\
	\vec{v} & \stackrel{L^+}{\longrightarrow} (\; 0, \ket{3}, 0, \ket{1}, 0, 0 \;),
	\notag \\
	\vec{v} & \stackrel{L^-}{\longrightarrow} (\; \ket{4}, 0, \ket{2}, 0, 0, 0 \;),
	\notag \\
	\vec{v} & \stackrel{P^+}{\longrightarrow} (\; 0, \ket{5}, 0, -i\ket{6}, \ket{1}, i\ket{3} \;),
	\notag \\
	\vec{v} & \stackrel{P^-}{\longrightarrow} (\; \ket{5}, 0, -i\ket{6}, 0, -\ket{2}, -i\ket{4} \;),
	\notag \\
	\vec{v} & \stackrel{Q^+}{\longrightarrow} (\; 0, \ket{6}, 0, i\ket{5}, -i\ket{3}, \ket{1} \;),
	\notag \\
	\vec{v} & \stackrel{Q^-}{\longrightarrow} (\; \ket{6}, 0, i\ket{5}, 0, i\ket{4}, -\ket{2} \;). 
	\label{ActionGenerators}
\end{align}
This defines the six-dimensional representation of $\mathfrak{k}. $ 

In this representation, the operators $L_x$ and $W$ are represented by the following matrices:
\begin{align}
	L_x &= \left(
	\begin{array}{cc:cc:c:c}
		0 & 1 & 0 & 0 & 0 & 0 \\
		\frac{1}{4} + \lambda^2 m & 0 & \lambda^2 v & 0 & -\lambda \alpha & -\lambda \beta \\[2pt] \hdashline
		0 & 0 & 0 & 1 & 0 & 0 \\
		\lambda^2 v & 0 & \frac{1}{4} + \lambda^2 m & 0 & -i\lambda \beta & i\lambda \alpha \\[2pt] \hdashline
		-\lambda \alpha & 0 & i\lambda \beta & 0 & 0 & 0 \\ \hdashline
		-\lambda \beta & 0 & -i\lambda \alpha & 0 & 0 & 0
	\end{array}
	\right),
	\notag \\[3pt]
	W &= \left(
	\begin{array}{cc:cc:c:c}
		c & a & h & f & \xi & \rho \\
		b & -c & g & -h & -\eta & -\sigma \\ \hdashline
		h & f & c & a & i\rho & -i\xi \\
		g & -h & b & -c & -i\sigma & i\eta \\ \hdashline
		-\eta & -\xi & i\sigma & i\rho & 0 & 0 \\ \hdashline
		-\sigma & -\rho & -i\eta & -i\xi & 0 & 0
	\end{array}
	\right).
\end{align}
These matrices have a block structure, with each block having a definite degree:
\begin{eqnarray}
	Z =
	\begin{pmatrix}
		A^{[00]} & A^{[11]} & A^{[10]} & A^{[01]} 
		\\
		B^{[11]} & B^{[00]} & B^{[01]} & B^{[10]} 
		\\
		C^{[10]} & C^{[01]} & C^{[00]} & C^{[11]} 
		\\
		D^{[01]} & D^{[10]} & D^{[11]} & D^{[00]} 
	\end{pmatrix}.
	\label{GradedMatrix}
\end{eqnarray}
We define the $\Z2$-graded trace of this matrix by \cite{rw2}
\begin{eqnarray}
	\Z2\text{-}\mathrm{tr}\, Z = \mathrm{tr}\, A^{[00]} + \mathrm{tr}\, B^{[00]} - \mathrm{tr}\, C^{[00]} - \mathrm{tr}\, D^{[00]}
\end{eqnarray}
where $\mathrm{tr}$ denotes the ordinary trace of a matrix. 

Now we define the Hamiltonian functional $ H(\lambda)$ of degree [00] as the functional whose left functional derivatives are given by
\begin{align}
	\frac{\delta^L H(\lambda)}{\delta m} &=\frac{1}{2} \Z2\text{-}\mathrm{tr} \left( W\frac{\partial L_x}{\partial m}\right) =  a \lambda^2,
	\notag\\
	\frac{\delta^L H(\lambda)}{\delta v} &=\frac{1}{2} \Z2\text{-}\mathrm{tr} \left( W\frac{\partial L_x}{\partial v}\right) = f \lambda^2,
	\notag\\
	\frac{\delta^L H(\lambda)}{\delta \alpha} &=-\frac{1}{2} \Z2\text{-}\mathrm{tr} \left( W \frac{\partial L_x}{\partial \alpha}\right)=  2\xi\lambda,
	\notag\\
	\frac{\delta^L H(\lambda)}{\delta \beta} &=-\frac{1}{2} \Z2\text{-}\mathrm{tr} \left( W\frac{\partial L_x}{\partial \beta}\right) = 2\rho\lambda.
\end{align}
Expanding $H(\lambda)$ as a Laurent series in $\lambda$, we write
\begin{align}
	H(\lambda) = \sum_{n \in \mathbb{Z}} H_{-2n}\lambda^{-2n}.
\end{align}
It follows that
\begin{alignat}{2}
	\frac{\delta^L H_{-2n}}{\delta m} &= a_{-2n-2},
	& \qquad \quad
	\frac{\delta^L H_{-2n}}{\delta v} &= f_{-2n-2},
	\notag\\
	\frac{\delta^L H_{-2n}}{\delta \alpha} &=  2\xi_{-2n-1},
	&
	\frac{\delta^L H_{-2n}}{\delta \beta} &= 2\rho_{-2n-1}. \label{HcompDef}
\end{alignat}
Combining these with  \eqref{recurrence}, the following proposition is immediate:
\begin{prop} \label{Prop:Rec1}
	The left functional derivatives of the Hamiltonian functionals satisfy the  recurrence relation:
	\begin{align}
		P_1 \begin{pmatrix}
			\frac{\delta^L H_{-2n-2} }{\delta m} \\[3pt]
			\frac{\delta^L H_{-2n-2}}{\delta v}\\[3pt]
			\frac{\delta^L H_{-2n-2}}{\delta \alpha}\\[3pt]
			\frac{\delta^L H_{-2n-2}}{\delta \beta}
		\end{pmatrix}
		=
		P_2 \begin{pmatrix}
			\frac{\delta^L H_{-2n}}{\delta m}\\[3pt]
			\frac{\delta^L H_{-2n}}{\delta v}\\[3pt]
			\frac{\delta^L H_{-2n}}{\delta \alpha}\\[3pt]
			\frac{\delta^L H_{-2n}}{\delta \beta}
		\end{pmatrix}.
		\label{RecRelP1P2}
	\end{align}
\end{prop}

It is also immediate from \eqref{EOM_P} that the $\Z2$-CH equation admits two distinct Hamiltonians:
\begin{align}
	&\begin{pmatrix}
		\dot{m}\\ \dot{v} \\ \dot{\alpha}\\ \dot{\beta}
	\end{pmatrix}
	= P_1
	\begin{pmatrix}
		\frac{\delta^L H_{-2}}{\delta m} \\[3pt] 
		\frac{\delta^L H_{-2}}{\delta v} \\[3pt] 
		\frac{\delta^L H_{-2}}{\delta \alpha} \\[3pt]
		\frac{\delta^L H_{-2}}{\delta \beta}
	\end{pmatrix}
	= P_2
	\begin{pmatrix}
		\frac{\delta^L H_0}{\delta m} \\[3pt]
		\frac{\delta^L H_0}{\delta v} \\[3pt]
		 \frac{\delta^L H_0}{\delta \alpha} \\[3pt]
		  \frac{\delta^L H_0}{\delta \beta}
	\end{pmatrix},
	\label{EoMbyP1P2}
\end{align}
where
\begin{align}
	H_{-2} &=-\frac12\int dx~ \left(m +\psi'\alpha +\chi'\beta\right),
	\notag\\
	H_{0} &=\frac{1}{2}\int dx \ \big[u(u-u'')+w(w-w'') \big].
	\label{Hamiltonians}
\end{align}

\subsection{Graded analogue of Hamiltonian functionals}

We seek Hamiltonian functionals that provide a nontrivially graded analogue of those defined in \eqref{HcompDef}. 
This can be done by looking for a $ 4 \times 4$ matrix $ N$ of the structure \eqref{GradedMatrix} satisfying
\begin{equation}
	[N, P_i] = 0, \quad i = 1, 2
\end{equation}
It is not difficult to verify that there is a unique such matrix, up to a scalar multiple, with nonvanishing entries only in the [11]-graded positions. It is given by
\begin{equation}
	N=
	\left(
	\begin{array}{cc:cc}
		& 1& & \\
		1& & & \\\hdashline
		& & & -i\\
		& & i& \\
	\end{array}
	\right),
	\qquad N^2 = \mathbb{I}_4,
\end{equation}  
where $ \mathbb{I}_4 $ denotes the $4 \times 4$ identity matrix. 

With this matrix $N$, we define the [11]-graded analogue $\tilde{H}_{-2n}$ of $ H_{-2n}$ by requiring its left functional derivatives to satisfy
\begin{equation}
	\frac{\delta^L \tilde{H}_{-2n}}{\delta m_i} = \sum_{j=1}^4 N_{ij}\frac{\delta^L H_{-2n}}{\delta m_j}. 
	\label{HtildDef}
\end{equation}
This relation implies that the left functional derivatives of both $\tilde{H}_{-2n}$ and $H_{-2n}$ are given by the solutions of the stationary zero-curvature equation:
\begin{align}
	\left(
	\begin{array}{c}
		a_{-2n-2} \\ f_{-2n-2} \\ 2\xi_{-2n-1} \\ 2\rho_{-2n-1}
	\end{array}
	\right)
	=
	\left(
	\begin{array}{c}
		\frac{\delta^L H_{-2n}}{\delta m}\\[4pt]
		\frac{\delta^L H_{-2n}}{\delta v}\\[4pt]
		\frac{\delta^L H_{-2n}}{\delta \alpha}\\[4pt]
		\frac{\delta^L H_{-2n}}{\delta \beta}
	\end{array}
	\right) 
	=
	\left(
	\begin{array}{c}
		\frac{\delta^L \tilde H_{-2n}}{\delta v}\\[4pt]
		\frac{\delta^L \tilde H_{-2n}}{\delta m}\\[4pt]
		-i\frac{\delta^L \tilde H_{-2n}}{\delta \beta}\\[4pt]
		i\frac{\delta^L \tilde H_{-2n}}{\delta \alpha}
	\end{array}
	\right).\label{relHH}
\end{align}
Using \eqref{HtildDef} and $ N^2 = \mathbb{I}_4$, we obtain a recurrence relation for $\tilde{H}_{-2n}$ from the recurrence relation for $H_{-2n}$ given in Proposition \ref{Prop:Rec1}.
\begin{prop}
	The left functional derivatives of $\tilde{H}_{-2n}$ satisfy the recurrence relation:
	\begin{align}
		P_1
		\left(
		\begin{array}{c}
			\frac{\delta^L \tilde H_{-2n-2}}{\delta m}\\[4pt]
			\frac{\delta^L \tilde H_{-2n-2}}{\delta v}\\[4pt]
			\frac{\delta^L \tilde H_{-2n-2}}{\delta \alpha}\\[4pt]
			\frac{\delta^L \tilde H_{-2n-2}}{\delta \beta}
		\end{array}
		\right) 
		=
		P_2
		\left(
		\begin{array}{c}
			\frac{\delta^L \tilde H_{-2n}}{\delta m}\\[4pt]
			\frac{\delta^L \tilde H_{-2n}}{\delta v}\\[4pt]
			\frac{\delta^L \tilde H_{-2n}}{\delta \alpha}\\[4pt]
			\frac{\delta^L \tilde H_{-2n}}{\delta \beta}
		\end{array}
		\right).
		\label{11rec3}
	\end{align}
\end{prop}
From the solution \eqref{EOM_P} of the stationary zero-curvature equation and \eqref{relHH}, we obtain the following explicit forms:
\begin{align}
	\tilde{H}_{0} =\frac{1}{2} \int dx \ [u(w-w'')+w(u-u'')],
	\notag\\
	\tilde{H}_{-2} = -\frac{1}{2} \int dx \ (v-i \ \psi' \beta + i \ \chi' \alpha).
\end{align}

\subsection{Poisson brackets and conservation laws}

We define a Poisson bracket on the space of $\Z2$-graded functionals. 
We consider equivalence classes of functionals modulo total $x$-derivatives:
\begin{eqnarray}
	\int dx\, a(x,\mathsf{m},\mathsf{m}',\dots) = \int dx\, (a + \partial_x h),
\end{eqnarray}
where the $\Z2$-graded functions $ a $ and $h$ are of the same degree. 
In the present case of the $\Z2$-CH equation, with a $4 \times 4$ matrix operator $\mathcal{J} = (\mathcal{J}_{ij})$ having the block structure \eqref{GradedMatrix}, the Poisson bracket of functionals $A$ and $B$ is defined by
\begin{align}
	\{A, B\}_{\mathcal{J}}&=\int dx \sum_{i,j}\frac{\delta^R A}{\delta m_i} \mathcal{J}_{ij} \frac{\delta^L B}{\delta m_j}  
	\notag \\ 
	&\overset{(\ref{RLfunctionalder})}{=}\int dx \sum_{i,j} (-1)^{(\hat A +\hat i)\cdot \hat i}~\frac{\delta^L A}{\delta m_i} \mathcal{J}_{ij} \frac{\delta^L B}{\delta m_j}. \label{DEFPB}
\end{align}

The $\Z2$-graded Poisson bracket is required to satisfy the following relations:
\begin{align}
	\{A, B\} &= -(-1)^{\hat A \cdot \hat B} \{B, A\},
	\label{SkewSym} \\
	(-1)^{\hat A \cdot \hat C}\{ \{A,B\}, C \} &+ (-1)^{\hat B \cdot \hat A}\{ \{B,C\}, A \} + (-1)^{\hat C \cdot \hat B}\{ \{C,A\},B \} = 0,
	\label{Jacobi}
\end{align}
where we assume that the Poisson bracket is a mapping of degree [00], which implies that $ \deg \mathcal{J}_{ij} = \hat i + \hat j.$  
The relations \eqref{SkewSym} and \eqref{Jacobi} are called the $\Z2$-skew-symmetry and the  $\Z2$-Jacobi relation, respectively. 
The requirements \eqref{SkewSym} and \eqref{Jacobi} impose restrictions on the matrix operator $ \mathcal{J}.$ 

\begin{thm} \label{Thm:Hpair}
	The pair $(P_1,P_2)$ given in \eqref{P1P2def} is a Hamiltonian pair. Namely, the Poisson brackets associated with $\mathcal{J}=P_1$, $\mathcal{J}=P_2$, and $\mathcal{J}=P_1+P_2$ satisfy the requirements \eqref{SkewSym} and \eqref{Jacobi}.
\end{thm}
\begin{proof}
	The proof is given in Appendix.
\end{proof}

In what follows, we consider only the Poisson brackets associated with $P_1$ and $P_2$. 
Time evolution of a functional of the momentum densities $ A[m_i] = \int dx \, a$ is given by the Hamiltonians \eqref{Hamiltonians}:
\begin{eqnarray}
	\frac{dA}{dt} = \{ A, H_{-2}\}_{P_1} = \{ A, H_{0}\}_{P_2}. \label{EvEqPoi}
\end{eqnarray}
\begin{proof}
	For the time evolution  determined by the pair $ (P_1, H_{-2})$,  
	we have
	\begin{align*}
		\frac{dA}{dt} &= \int dx \sum_{i}\dot{m}_i \frac{\delta^L A}{\delta m_i}
		\overset{\eqref{EoMbyP1P2}}{=} 
		\int dx \sum_{i,j}(P_1)_{ij} \frac{\delta^L H_{-2}}{\delta m_j} \frac{\delta^L A}{\delta m_i} 
		\\
		&= \int dx \sum_{i,j}(-1)^{\hat{i}\cdot(\hat{i}+\hat{A})} \frac{\delta^L A}{\delta m_i} (P_1)_{ij} \frac{\delta^L H_{-2}}{\delta m_j}
		\overset{\eqref{DEFPB}}{=} \{ A, H_{-2} \}_{P_1}.
	\end{align*}
	One can repeat the same computation for the pair $(P_2, H_0).$
\end{proof}

With the relation \eqref{EvEqPoi}, we have established the bi-Hamiltonian structure of the $\Z2$-CH equation. 
Consequently, the $\Z2$-CH equation possesses infinitely many involutive conserved charges. 

\begin{thm} \label{Thm:Conservation}
	The Hamiltonian functionals $H_{-2n}$ defined in \eqref{HcompDef} and $\tilde{H}_{-2n}$ defined in  \eqref{HtildDef} satisfy the following properties:
	\begin{enumerate}
		\item They are conserved:
		    \begin{equation}
		    	\frac{d H_{-2n}}{dt} = \frac{d \tilde{H}_{-2n}}{dt} = 0,
		    \end{equation}
		\item They are in involution:
		
		    \begin{align}
		    	 \{ H_{-2n}, H_{-2m}\}_{P_k}= \{\tilde H_{-2n}, \tilde H_{-2m}\}_{P_k}= 
		    	 \{ H_{-2n}, \tilde H_{-2m}\}_{P_k}= 0, \quad k = 1, 2
		    \end{align}
	\end{enumerate}
\end{thm}
To prove Theorem \ref{Thm:Conservation}, we need the following lemma:
\begin{lemma}
	$H_{-2n}$ and $ \tilde{H}_{-2n}$ satisfy the recurrence relations:
	\begin{align}
		\{H_{-2n}, H_{-2m}\}_{P_1} &=\{H_{-2n}, H_{-2m+2}\}_{P_2}=\{H_{-2n-2}, H_{-2m+2}\}_{P_1},
		\label{RR1} \\
		\{\tilde H_{-2n}, \tilde H_{-2m}\}_{P_1} &=\{\tilde H_{-2n}, \tilde H_{-2m+2}\}_{P_2}=\{\tilde H_{-2n-2}, \tilde H_{-2m+2}\}_{P_1},
		\label{RR2} \\
		\{H_{-2n}, \tilde H_{-2m}\}_{P_1} &=\{H_{-2n}, \tilde H_{-2m+2}\}_{P_2}=\{H_{-2n-2}, \tilde H_{-2m+2}\}_{P_1}.
		\label{RR3}
	\end{align}
\end{lemma}
\begin{proof}
	We prove only \eqref{RR1} here. The other relations follow similarly. 
	The first equality in \eqref{RR1} is proved as follows:
	\begin{align*}
		\{H_{-2n}, H_{-2m}\}_{P_1} &= \int dx \sum_{i,j} \frac{\delta^R H_{-2n}}{\delta m_i}(P_1)_{ij} \frac{\delta^L H_{-2m}}{\delta m_j}
		\\
		&\overset{\eqref{RecRelP1P2}}{=} \int dx \sum_{i,j} \frac{\delta^R H_{-2n}}{\delta m_i}(P_2)_{ij} \frac{\delta^L H_{-2m+2}}{\delta m_j}
		\\
		&= \{ H_{-2n}, H_{-2m+2} \}_{P_2}.
	\end{align*}
	The second equality in \eqref{RR1} is proved as follows:
	\begin{align*}
		\{ H_{-2m+2}, H_{-2n} \}_{P_2} &= 
		\int dx \sum_{i,j} \frac{\delta^R H_{-2m+2}}{\delta m_i} (P_2)_{ij}\frac{\delta^L H_{-2n}}{\delta m_j}
		\\
		&\overset{\eqref{RecRelP1P2}}{=} \int dx \sum_{i,j}\frac{\delta^R H_{-2m+2}}{\delta m_i}(P_1)_{ij} \frac{\delta^L H_{-2n-2}}{\delta m_j}
		\\
		&= \{ H_{-2m+2}, H_{-2n-2} \}_{P_1}.
	\end{align*}
	Then, using the skew-symmetry of the Poisson bracket, we obtain \eqref{RR1}.
\end{proof}

\noindent
\textit{Proof of Theorem \ref{Thm:Conservation}.} 
Repeated use of \eqref{RR1} shows that, for any integers $n, m$, there exists an integer $k$ such that
\begin{equation}
	\{ H_{-2n}, H_{-2m} \}_{P_2} = 
	\begin{cases}
		\{ H_{-2k}, H_{-2k}  \}_{P_2}, &\quad n-m \equiv 0 \mod 2
		\\[10pt]
		\{ H_{-2k}, H_{-2k}  \}_{P_1}, &\quad n-m \equiv 1 \mod 2
	\end{cases}
\end{equation}
The right-hand side is identically zero due to the $\Z2$-skew-symmetry of the Poisson brackets. 
Therefore, we have $ \{ H_{-2n}, H_{-2m} \}_{P_2} = \{ H_{-2n}, H_{-2m-2} \}_{P_1} = 0 $ for any $n, m$. Similarly, using \eqref{RR2} one can show that $ \{ \tilde{H}_{-2n}, \tilde{H}_{-2m} \}_{P_k} = 0$ for $k=1,2.$ 
The conservation of $H_{-2n}$ then follows immediately. 

Next, by the repeated use of \eqref{RR3} one can see that, for any integers $n, m$, there exists $ k \in \mathbb{Z}$ such that
\begin{equation}
	\{H_{-2n}, \tilde{H}_{-2m} \}_{P_2} = 
	\begin{cases}
		\{H_{-2k}, \tilde{H}_{-2k} \}_{P_1}, & \quad n-m \equiv 1 \mod 2
		\\[10pt]
		\{H_{-2k}, \tilde{H}_{-2k} \}_{P_2}, & \quad n-m \equiv 0 \mod 2
	\end{cases}
\end{equation}
The right-hand side vanishes identically, as we show below. 
For  the sake of simplicity, we write
\begin{equation*}
	h_i^R := \frac{\delta^R H_{-2k}}{\delta m_i},
	\qquad
	h_i^L := \frac{\delta^L H_{-2k}}{\delta m_i},
	\qquad
	\tilde{h}_i^R := \frac{\delta^R \tilde{H}_{-2k}}{\delta m_i},
	\qquad
	\tilde{h}_i^L := \frac{\delta^L \tilde{H}_{-2k}}{\delta m_i}.
\end{equation*}
By the explicit form of $ P_2$ given in \eqref{P1P2def} and $ \displaystyle \tilde{h}_i^L = \sum _j N_{ij} h_j^L, $ we have
\begin{align*}
	\{H_{-2k}, &\tilde{H}_{-2k} \}_{P_2} \\
	&= 
	\int dx \sum_{i,j} h^R_i(P_2)_{ij} \tilde{h}^L_j 
	\\
	&= \int dx \big(  h^R_1 M h^L_2 + h^R_1 V h^L_1 + h^R_2 V h^L_2 + h^R_2 M h^L_1 - \frac{i}{2} h^R_3 L h^L_4 + \frac{i}{2} h^R_4 L h^L_3 \big).
\end{align*}
By integration by parts, the operators $M$, $V$, and $L$ can be transferred to act on $h^R_i$. 
We then express $h^R_i$ in terms of $h^L_i$ using \eqref{RLfunctionalder}, and vice versa.  
Interchanging the factors in each term, we obtain 
\begin{align*}
	\{H_{-2k}, &\tilde{H}_{-2k} \}_{P_2}\\
	&=
	-\int dx \big(  h^R_2 M h^L_1 + h^R_1 V h^L_1 + h^R_2 V h^L_2 + h^R_1 M h^L_2 + \frac{i}{2} h^R_4 L h^L_3 - \frac{i}{2} h^R_3 L h^L_4  \big)
	\\
	&=-\{H_{-2k},\tilde H_{-2k}\}_{P_2},
\end{align*}
which implies that $ \{H_{-2k}, \tilde{H}_{-2k} \}_{P_2} = 0.$ 
Similarly, one can show that $ 	\{H_{-2k}, \tilde{H}_{-2k} \}_{P_1} = 0.$ 
Therefore, we have shown that $ \{H_{-2n}, \tilde{H}_{-2m} \}_{P_2} = \{H_{-2n}, \tilde{H}_{-2m-2} \}_{P_1} =0 $ for any $n, m$. 
The conservation of $\tilde{H}_{-2n}$ then follows immediately. \hfill $\Box$

\section{Concluding remarks} \label{SEC:CR}

By constructing a Lax pair in the loop algebra of $\Z2$-$\osp(1|2)$, we have derived a $\Z2$-graded extension of the CH equation. 
The $\Z2$-CH equation has one dynamical variable in each graded subspace and belongs to a hierarchy determined by a particular choice of $L_x$ given in \eqref{LaxOp}. 
With the aid of the stationary zero-curvature equation, we have established the bi-Hamiltonian structure for the $\Z2$-CH equation. 
Remarkably, there exist infinitely many conserved quantities of degree [11], in addition to the ones of degree [00]. 
All of these conserved quantities are in involution. 
In this paper, we have focused only on Poisson brackets of degree [00]. 
One may also consider Poisson brackets of degree [11], which would allow the time evolution of the system to be expressed in terms of the [11]-graded Hamiltonian functionals.

The present construction of the bi-Hamiltonian structure and conserved charges should also be applicable to the $\Z2$-graded extension of the KdV and mKdV equations introduced in \cite{AiFujiIto}, which were also derived from the loop algebra of $\Z2$-$\osp(1|2)$ in the same manner as in the present work.

Here, we mention two important open problems concerning the $\Z2$-graded extension of the CH, KdV and mKdV equations.
The first is to clarify the conformal nature of these equations in terms of their underlying Virasoro algebraic structure and Hamiltonian mechanics on the Bott-Virasoro group. 
A $\Z2$-graded extension of the Virasoro algebra based on $\Z2$-$\osp(1|2)$ was introduced in \cite{AFITTsuper,AiSe}. However, a dynamical realization of this $\Z2$-graded Virasoro algebra has not yet been found. 
The second is to develop a classical $r$-matrix formulation of these equations and clarify its relation to their bi-Hamiltonian and integrable structures. 
To the best of our knowledge, a $\Z2$-graded extension of the classical Yang-Baxter equation has not yet been discussed in the literature.

\section*{Acknowledgements}

N. A. is supported by JSPS KAKENHI Grant Number JP23K03217. 
R. I. is supported by JST SPRING, Grant Number JPMJSP2139. 
L. V. was financed in part by the Coordena\c{c}\~{a}o de Aperfei\c{c}oamento de Pessoal de N\'{i}vel Superior - Brasil (CAPES) - Finance Code 001.

\section*{Appendix : Proof of Theorem \ref{Thm:Hpair}}

\setcounter{section}{1}
\renewcommand{\thedfn}{\Alph{section}\arabic{dfn}}
\renewcommand{\thethm}{\Alph{section}\arabic{thm}}
\renewcommand{\theprop}{\Alph{section}\arabic{prop}}
\renewcommand{\thelemma}{\Alph{section}\arabic{lemma}}

\renewcommand{\thesection}{\Alph{section}}
\setcounter{equation}{0}

\setcounter{dfn}{0}
\setcounter{thm}{0}
\setcounter{prop}{0}
\setcounter{lemma}{0}

Theorem \ref{Thm:Hpair} can be proved in essentially the same way as in Chapter 7 of \cite{Olver1993applications}.
The only difference lies in how the $\Z2$-grading is incorporated into the proof. 
Therefore, rather than repeating the proof in detail, we indicate how the propositions and theorems in Chapter 7 of \cite{Olver1993applications} are adapted to the $\Z2$-graded setting.

Let $m_1, m_2, \dots, m_q$ be $\Z2$-graded functions, which are not necessarily the momentum densities of the $\Z2$-CH equation. We consider a more general setting here. 
We denote $\hat{i} := \deg m_i.$ 
The Poisson bracket of degree [00] is defined by the relation \eqref{DEFPB}.

First, to prove the $\Z2$-skew-symmetry, we extend the notion of the adjoint of the matrix operator $\mathcal{J}$ in \eqref{DEFPB} to $\Z2$-graded setting. 
Let $(f_1, f_2, \dots, f_q)$ and $  (g_1, g_2, \dots, g_q)$ be vectors of degree $\hat{f}$ and $ \hat{g}$, respectively. 
By this, we mean that their components are $\Z2$-graded functions of degrees $ \deg f_i = \hat{f} + \hat{i}$ and $ \deg g_i = \hat{g} + \hat{i}.$
\begin{dfn}
	The adjoint $\mathcal{J}^{\dagger}$ of $\mathcal{J}$ is defined by
	\begin{equation}
		\int dx \sum_{i,j} f_i \mathcal{J}_{ij} g_j = \int dx \sum_{i,j} (-1)^{\hat{f_i}\cdot(\hat i+\hat j)}  \Big((\mathcal{J}^\dag)_{ji}f_i \Big)  g_j \label{adjointj}. 
	\end{equation}
	The matrix operator $\mathcal{J}$ is is said to be skew-adjoint if $(\mathcal{J}^\dag)_{ij}=-(-1)^{\hat i \cdot \hat j} \mathcal{J}_{ij}$. 
\end{dfn}

\noindent 
\textbf{Remark.} The entries of $\mathcal{J}^{\dagger}$ are given by
\begin{equation}
	(\mathcal{J}^\dagger) _{ij}=(\mathcal{J}_{ji})^\dagger. \label{adjointelements} 
\end{equation} 
This follows by applying \eqref{DEF:ajoint} to the left-hand side of \eqref{adjointj}. 
With this, it is easy to verify the following:
\begin{prop}
	If the matrix operator $\mathcal{J}$ is skew-adjoint, then the corresponding Poisson bracket is $\Z2$-skew-symmetric:
	\[
	\{A,B\}_{\mathcal{J}}=-(-1)^{\hat A \cdot \hat B}\{B,A\}_{\mathcal{J}}.
	\]
\end{prop}

For the operators $P_1 $ and $P_2$ of the $\Z2$-CH equation, one can see that they are skew-adjoint from the following computations of their adjoint:
\begin{align}
	P_1 ^\dagger &\overset{\eqref{adjointelements}}{=}    \begin{pmatrix}
		p_{1}^\dagger &p_{2}^\dagger &p_{5}^\dagger &ip_{6}^\dagger\\
		p_{2}^\dagger &p_{1}^\dagger &p_{6}^\dagger &ip_{5}^\dagger\\
		p_{3}^\dagger &ip_{4}^\dagger &p_{7}^\dagger &-p_{8}^\dagger\\
		p_{4}^\dagger &-ip_{3}^\dagger &p_{8}^\dagger &p_{7}^\dagger
	\end{pmatrix}
	\overset{\eqref{DEF:ajoint}}{=}
	\begin{pmatrix}
		-p_{1} &-p_{2} &-p_{3} &-p_{4}\\
		-p_{2} &-p_{1} &ip_{4} &-ip_{3}\\
		-p_{5} &p_{6} &p_{7} &-p_{8}\\
		-ip_{6} &ip_{5} &p_{8} &p_{7}
	\end{pmatrix},
	\notag \\[3pt]
	P_2^\dagger &\overset{\eqref{adjointelements}}{=} \begin{pmatrix}
		M^\dagger &V^\dagger &0 &0\\
		V^\dagger &M^\dagger &0 &0\\
		0 &0 &\frac{1}{2}L^\dagger &0\\
		0 &0 &0 &\frac{1}{2}L^\dagger
	\end{pmatrix}
	\overset{\eqref{DEF:ajoint}}{=}
	\begin{pmatrix}
		-M &-V &0 &0\\
		-V &-M &0 &0\\
		0 &0 &\frac{1}{2}L &0\\
		0 &0 &0 &\frac{1}{2}L
	\end{pmatrix}.
\end{align}
Thus, the Poisson brackets associated with $P_1$ and $P_2$ are $\Z2$-skew-symmetric.  

Next, we consider the $\Z2$-Jacobi relation for the Poisson bracket associated with a skew-adjoint operator $P$.  
Throughout this discussion, $P$ is assumed to be skew-adjoint. 
Let $ F = (F_1, F_2, \dots, F_q)$ be a vector of degree $\hat{F}$, i.e., $ \hat{F}_i = \hat{F} + \hat{i}.$  
We define the evolutionary vector field associated with $F$ by the right derivative as
\begin{equation}
	\overleftarrow{\mathbf{V}}_F = \sum_i \overleftarrow{\frac{\partial}{\partial m_i}}F_i.
\end{equation}
To consider the prolongation of this vector field, we introduce the total derivative operator defined using the left derivative as
\begin{equation}
	D_x = \partial_x + \sum_{i,n} m_i^{(n+1)} \overrightarrow{\frac{\partial }{\partial m_i^{(n)}}},
\end{equation}
where $ m_i^{(n)} := \partial_x^n\, m_i$ as before. 
For better readability,  we do not indicate the arrow over $D_x$. 
The prolongation of  $ \overleftarrow{\mathbf{V}}_F $ is then defined by
\begin{equation}
	\pr{F} =\sum_{i,n} \overleftarrow{\frac{\partial}{\partial m_i ^{(n)}}} \big[(D_x)^n F_i \big]. 
	\label{Def:Prolongation}
\end{equation} 

Verification of the $\Z2$-Jacobi relation is simplified significantly by introducing the dual objects $dm_i^{(n)}$ to $  \overrightarrow{\frac{\partial }{\partial m_i^{(n)}}}. $ 
The duality pairing is defined by 
\begin{equation}
	\left\langle dm_i^{(n)}, \overleftarrow{\frac{\partial}{\partial m_j^{(r)}}} \right\rangle = \delta_{ij} \delta^{nr}.
\end{equation}
This pairing is extended linearly as
\begin{equation}
	\left\langle \sum_{i,n} \alpha_{i,n}dm_i^{(n)}, \sum_{j,r} \overleftarrow{\frac{\partial}{\partial m_j^{(r)}}} X_{j,r} \right\rangle = \sum_{i,n} \alpha_{i,n} X_{i,n},
\end{equation}
and to tensor product as
\begin{equation}
	\left\langle   dm_i^{(n)} \otimes dm_j^{(r)},  \overleftarrow{\frac{\partial}{\partial m_k^{(s)}}} \otimes \overleftarrow{\frac{\partial}{\partial m_{l}^{(t)}}} \right\rangle
	= 
	(-1)^{\hat{j}\cdot \hat{k}}
	\left\langle  dm_i^{(n)}, \overleftarrow{\frac{\partial}{\partial m_k^{(s)}}} \right\rangle
	\,
	\left\langle  dm_j^{(r)}, \overleftarrow{\frac{\partial}{\partial m_l^{(t)}}} \right\rangle.
\end{equation}

We introduce the functional bi-vector $\Theta_P$ associated with  $P$ as
\begin{align}
	\Theta_P =- \frac{1}{2}\int dx \sum_{i,j} P_{ij}\, dm_j \wedge dm_i,
\end{align}
where the $\Z2$-graded wedge product is defined by
\begin{equation}
	dm_i^{(n)} \wedge dm_j^{(r)} := dm_i^{(n)} \otimes dm_j^{(r)}  - (-1)^{\hat{i}\cdot \hat{j}} dm_j^{(r)} \otimes dm_i^{(n)}. 
\end{equation}
With these settings, we obtain the following proposition:
\begin{prop} \label{Prop:Tri}
	If $ dm_i^{(n)} \pr{Pdm} = 0$ for all $i$ and $n$, then 
	the $\Z2$-Jacobi relation is equivalent to 
	\begin{equation}
		\left\langle  \Theta_P \pr{Pdm} ,  \pr{L(a)} \otimes \pr{L(b)} \otimes \pr{L(c)}\right\rangle = 0, \label{Tri}
	\end{equation}
	where $ P dm $ and $L(a)$ are vectors whose $i$-th components are given by $\displaystyle \sum_j P_{ij}\, dm_j$ and $L_i(a)$, respectively. Here $L_i(a)$ denotes the left Euler operator \eqref{DefEuler} for a functional $ A = \int dx\; a$. 
\end{prop}
The following theorem is an immediate consequence of Proposition \ref{Prop:Tri}:
\begin{thm}
	The Poisson bracket associated with $P$ satisfies the $\Z2$-Jacobi relation if and only if
	\begin{equation}
		\Theta_P \pr{Pdm} = 0. \label{Trivector}
	\end{equation}
\end{thm}
We call a skew-adjoint operator $P$ a \textit{Hamiltonian operator} if it defines a Poisson bracket. 

\begin{prop}
	Let $P $ and $Q$ be Hamiltonian operators. Then $ P+Q$ is also a Hamiltonian operator if and only if
	\begin{align}
		(\Theta_P + \Theta_Q ) \pr{(P+Q)dm} =0. \label{PplusQ}
	\end{align}
\end{prop}

For the $\Z2$-CH equation, a direct computation shows that the skew-adjoint operators $P_1$ and $P_2$ satisfy \eqref{Trivector} and \eqref{PplusQ}. 
This completes the proof of Theorem \ref{Thm:Hpair}. 

As we have seen above, Proposition \ref{Prop:Tri} plays an essential role in the proof of Theorem \ref{Thm:Hpair}. In what follows, we outline the proof of Proposition \ref{Prop:Tri}, focusing on the modifications required to adapt the proof in Chapter 7 of \cite{Olver1993applications} to the $\Z2$-graded setting.

First, we observe that the Poisson bracket \eqref{DEFPB} can be expressed in terms of the prolongation:
\begin{align}
	\{A,B\}_P= A \pr{PL(b)}. \label{PBprolong}
\end{align} 
This follows by evaluating the right-hand side and using Definition \ref{Def:FuncDiff}. 
We then introduce the operator $ \FD{jk}{b} $ associated with a $\Z2$-graded function $b$, defined by  
\begin{align}
	\FD{jk}{b}  = \sum_{n,r} \left(\overrightarrow{\frac{\partial}{\partial m_k^{(n)}}}(-D_x)^r \overrightarrow{\frac{\partial }{\partial m_j^{(r)}}} \,b  \right)D_x^n,
\end{align}
where $D_x^n := D_x \cdot D_x \cdots D_x$ ($n$ times). 
Note that $\deg(\FD{jk}{b})= \hat j+ \hat k +\hat b.$ 
This operator is a $\Z2$-graded analogue of the Fr\'echet derivative. 
\begin{prop} \label{Prop:Frechet}
	The operator $\FD{jk}{b}$ satisfies the following properties:
	\begin{align}
		L_j(b) \pr{PL(c)} &= \sum_{k,l} (-1)^{(\hat b +\hat j +\hat k)\cdot \hat k} \,\FD{jk}{b} (P_{kl}L_l(c)), \label{Property1}
		\\
		\big(\FD{kj}{b} \big)^{\dagger} &= (-1)^{\hat k \cdot \hat j}\, \FD{jk}{b}.
		\label{Property2}
	\end{align}
	
\end{prop}
\begin{proof}
	Using the relation between the left and right derivatives, \eqref{Property1} follows by direct computation. To prove \eqref{Property2}, we use Lemma \ref{Lemma:Frechet}, given below.  
	Applying Lemma \ref{Lemma:Frechet} repeatedly, we obtain, for any $\Z2$-graded functions $ F_j$ and $ G_k$,
	\begin{align*}
		\int dx \sum_{j,k}  F_j\, \FD{jk}{b}\, G_k&= \int dx \sum_{j,k,n,r} F_j\left( \overrightarrow{\frac{\partial}{\partial m^{(n)}_k}} (-D_x)^r \overrightarrow{\frac{\partial}{\partial m_j ^{(r)}}} \, b\right)D_x^n G_k\\
		&=\int dx  \sum_{j,k,n,r} F_j (-D_x)^r 
		\left[ \left(\overrightarrow{\frac{\partial}{\partial m_k^{(n)}}} \; \overrightarrow{\frac{\partial}{\partial m_j^{(r)}}} \, b \right)D_x^n G_k
		\right]. 
	\end{align*}
	Repeating the integration by parts, we can rewrite this as
	\begin{align*}
		&\int dx \sum_{j,k,n,r}  (D_x^r F_j)  \left((-1)^{\hat k \cdot \hat j}\overrightarrow{\frac{\partial}{\partial m_j^{(r)}}} \;
			\overrightarrow{ \frac{\partial}{\partial m_k^{(n)}}} \,b \right)D_x^nG_k 
		\\
		&= \int dx \sum_{j,k,n,r} (-1)^{\hat F_j \cdot (\hat j + \hat k + \hat b) + \hat k \cdot \hat j} \left( (-D_x)^n \left[ \left( \overrightarrow{\frac{\partial}{\partial m_j^{(r)}}} \; \overrightarrow{\frac{\partial}{\partial m_k^{(n)}}}\,b\right) \Big(D^r_x F_j \Big) 
		\right] \right) G_k
		\\
		&\overset{\eqref{derivativelemma}}{=} \int dx \sum_{j,k} (-1)^{\hat F_j \cdot  (\hat j + \hat k + \hat b) } \left((-1)^{\hat k \cdot \hat j} \FD{kj}{b} F_j  \right)G_k.
	\end{align*}
This completes the proof of \eqref{Property2}.
\end{proof}

\begin{lemma} \label{Lemma:Frechet}
	For any functions $P, Q$ of $ m_i^{(n)}$, we have
	\begin{align}
		\sum_n \left(\overrightarrow{ \frac{\partial}{\partial m_j^{(n)} } } D_x Q\right) D_x^n P = D_x \left[
	    	\sum_n \left( \overrightarrow{ \frac{\partial}{\partial m_j^{(n)}} }Q \right)D_x^n P 
		\right]. \label{derivativelemma}
	\end{align}
\end{lemma}
\begin{proof}
	Applying the commutation relation:
	\begin{equation}
		\left[\frac{\overrightarrow{\partial}}{\partial m_j^{(n)}},D_x \right]= \frac{\overrightarrow{\partial}}{\partial m_j ^{(n-1)}}
	\end{equation}
	to the left-hand side, we have
	\[
	   \mathrm{l.h.s.} 
	   = 
	   \sum_n \left[
	      \left( D_x \,\overrightarrow{\frac{\partial}{\partial m_j^{(n)}}}\, Q \right) D_x^n P + 
	      \left(\overrightarrow{\frac{\partial}{\partial m_j^{(n)}}}\, Q \right) D_x^{n+1}P
	   \right].
	\]
	This is identical to the right-hand side.
\end{proof}

We now compute the first term in the $\Z2$-Jacobi relation. By the Leibniz rule for the prolongation \eqref{Def:Prolongation}, we proceed as follows:
\begin{align}
	\{\{A,B\}_P,C\}_P&\overset{\eqref{PBprolong}}{=}  \{A,B\}_P \pr{PL(c)}    
	\notag\\
	&\overset{\eqref{DEFPB}}{=}
	\int dx \sum_{i,j} \big( R_i(a)P_{ij}L_j(b) \big) \pr{PL(c)}
	\notag \\
	&=\int dx \sum_{i,j} \Big[ 
	R_i(a) P_{ij} \;\big(  L_j(b) \pr{PL(c)} \big)+ (-1)^{\hat c\cdot (\hat b + \hat j)}R_i(a)  \big( P_{ij} \pr{PL(c)}  \big) L _j(b)\notag 
	\notag\\
	&+(-1)^{\hat c \cdot (\hat b +\hat i)}\; 
	\big( R_i(a) \pr{PL(c)} \big)
	P_{ij}L_j(b)
	\Big].
	\label{JacobiFirst}
\end{align}
The differential operator $P_{ij}$ is written as
\begin{equation}
	P_{ij} = \sum_n p_{ij,n}D_x^n
\end{equation}
where $p_{ij,n}$ are some $\Z2$-graded functions. 
In the second term of \eqref{JacobiFirst}, $P_{ij}\pr{PL(c)}$ should be understood as a left differential operator:
\begin{equation}
	P_{ij}\pr{PL(c)}
	=\sum_n p_{ij,n}\pr{PL(c)}D_x^n.
\end{equation}

We can rewrite the first term in \eqref{JacobiFirst} by the following lemma:
\begin{lemma}
	The following identity holds:
	\begin{align}
		\int dx \sum_{i,j} R_i (a) P_{ij} \; \Big(  L_j(b) \pr{PL(c)} \Big)=-\int dx \sum_{i,j} (-1)^{\hat a \cdot ( \hat b +\hat i)} \Big( R_i(b) \pr{PL(a)} \Big) P_{ij}L_j(c). \label{z2z2prolong}
	\end{align}
\end{lemma}
\begin{proof}
	Using Proposition \ref{Prop:Frechet}, the identity follows by direct computation. 
\end{proof}

We repeat the same computation for the second and third terms in $\Z2$-Jacobi relation, and sum all three terms. We then obtain
\begin{thm}
	The $\mathbb{Z}^2_2$-Jacobi identity \eqref{Jacobi} is equivalent to
	\begin{align}
		\int dx \sum_{i,j} \left[\; (-1)^{\hat c \cdot (\hat a + \hat b + \hat j)} R_i(a) \big( P_{ij}  \pr{PL(c)} \big)  L_j(b) + \mathrm{cyclic\ perm.} \; \right]=0. \label{cyclicperm}
	\end{align}
\end{thm}
A direct computation of \eqref{Tri} shows that it is identical to \eqref{cyclicperm}, which establishes  Proposition \ref{Prop:Tri}. 

\bibliographystyle{JHEP} 
\bibliography{Integrable}
\end{document}